\documentclass[11pt]{article}

\usepackage[T1]{fontenc}
\usepackage[utf8]{inputenc}
\usepackage{lmodern}
\usepackage{amsmath,amssymb,amsfonts,amsthm,mathtools}
\usepackage{aliascnt}
\usepackage{bm}
\usepackage{graphicx}
\usepackage{microtype}
\usepackage{enumitem}
\usepackage[backend=biber,style=numeric-comp,sorting=none,maxbibnames=99]{biblatex}
\usepackage[colorlinks=true,linkcolor=blue,citecolor=blue,urlcolor=blue]{hyperref}
\usepackage[nameinlink,noabbrev]{cleveref}

\crefname{theorem}{theorem}{theorems}
\Crefname{theorem}{Theorem}{Theorems}
\crefname{proposition}{proposition}{propositions}
\Crefname{proposition}{Proposition}{Propositions}
\crefname{lemma}{lemma}{lemmas}
\Crefname{lemma}{Lemma}{Lemmas}
\crefname{corollary}{corollary}{corollaries}
\Crefname{corollary}{Corollary}{Corollaries}
\crefname{definition}{definition}{definitions}
\Crefname{definition}{Definition}{Definitions}
\crefname{example}{example}{examples}
\Crefname{example}{Example}{Examples}
\crefname{remark}{remark}{remarks}
\Crefname{remark}{Remark}{Remarks}

\graphicspath{{figures/}}

\newtheorem{theorem}{Theorem}[section]

\newaliascnt{proposition}{theorem}
\newtheorem{proposition}[proposition]{Proposition}
\aliascntresetthe{proposition}

\newaliascnt{lemma}{theorem}
\newtheorem{lemma}[lemma]{Lemma}
\aliascntresetthe{lemma}

\newaliascnt{corollary}{theorem}
\newtheorem{corollary}[corollary]{Corollary}
\aliascntresetthe{corollary}

\theoremstyle{definition}
\newaliascnt{definition}{theorem}
\newtheorem{definition}[definition]{Definition}
\aliascntresetthe{definition}

\newaliascnt{example}{theorem}
\newtheorem{example}[example]{Example}
\aliascntresetthe{example}

\newaliascnt{remark}{theorem}
\newtheorem{remark}[remark]{Remark}
\aliascntresetthe{remark}

\newcommand{\Hil}{\mathcal H}
\newcommand{\Fock}{\mathcal F}
\newcommand{\Proj}{\mathbb P}
\newcommand{\CP}{\mathbb{CP}}
\newcommand{\C}{\mathbb C}
\newcommand{\Zmap}{\mathcal Z}
\newcommand{\Conf}{\mathrm{Conf}}
\newcommand{\Disc}{\mathrm{Disc}}
\newcommand{\Sr}{\mathcal S_r^{\mathrm{reg}}}

\newcommand{\Vieta}{\mathcal V}
\newcommand{\Div}{\operatorname{Div}}
\newcommand{\ket}[1]{|#1\rangle}

\title{Stellar Braid Monodromy of Finite-Rank Non-Gaussian Photonic States}
\author{
	Arnaud Coatanhay and Angélique Drémeau\\
	Lab-STICC, UMR CNRS 6285, ENSTA, Institut Polytechnique de Paris,\\
	2 rue Fran\c{c}ois Verny, 29806 Brest Cedex 9, France\\
	\small \texttt{arnaud.coatanhay@ensta.fr} and \texttt{angelique.dremeau@ensta.fr}
}
\date{\today}

\begin{document}

\maketitle

\begin{abstract}
Finite-rank non-Gaussian bosonic states admit a holomorphic description in the Bargmann representation: after a zero-free Gaussian factor is separated off, their non-Gaussian structure is encoded by a finite stellar divisor. This article introduces a topological refinement of stellar rank for regular parameterized families of such states. Rather than only counting the zeros of the stellar divisor, we follow their motion under deformations of the state and record the associated braid monodromy. In the finite-Fock chart, a regular degree-$r$ stellar state is represented by a monic polynomial with $r$ simple zeros. The regular stratum is biholomorphic to the unordered configuration space $\Conf_r(\C)/S_r$, and its fundamental group is the Artin braid group $B_r$. Thus braid monodromy is an intrinsic invariant of loops in the regular finite-rank stellar state space. We then extend the construction to admissible finite stellar divisors $E_{\tau,\mu}(z)P(z)$; the zero-free Gaussian parameters form a contractible fiber over the same configuration-space base. Experimentally motivated finite-Fock families, especially the cubic subspace spanned by $\ket{0},\ldots,\ket{3}$, provide concrete laboratories, while trinomial slices yield explicit discriminants and local half-twists. The resulting invariant is post-tomographic and applies to preparation loops and parameterized families; it complements Wigner negativity, stellar rank, approximate stellar rank, and other scalar diagnostics of non-Gaussianity.
\end{abstract}

\section{Introduction}

Gaussian states form the baseline geometry of continuous-variable quantum optics. They are determined by first and second moments, transformed by quadratic Hamiltonians, and described naturally by symplectic methods \cite{BraunsteinVanLoock2005,Weedbrook2012}. This structure is powerful, but it also marks the boundary of the Gaussian regime. Universal continuous-variable quantum computation and many genuinely quantum information tasks require non-Gaussian resources \cite{LloydBraunstein1999,Bartlett2002,Walschaers2021}.

The difficulty is that non-Gaussianity is not a single scalar feature. Wigner negativity is an important diagnostic: for pure continuous-variable states, positivity of the Wigner function characterizes Gaussianity, and the negative volume is a standard nonclassicality indicator \cite{Hudson1974,KenfackZyczkowski2004}. Yet two states may have the same coarse non-Gaussian diagnosis while differing in the geometric organization of their higher-order structure. For this reason, it is useful to ask not only how much non-Gaussianity is present, but also how it is arranged and how it changes under controlled deformations.

This question has become experimentally relevant. Finite superpositions of traveling-field photon-number states can be engineered by conditional operations \cite{Dakna1999}, and arbitrary optical superpositions up to three photons have been generated and reconstructed tomographically \cite{Yukawa2013}. More recent work addresses optical states of higher stellar rank and practical certification requirements \cite{Provaznik2025Optical,Provaznik2026Witnesses}. In another bosonic platform, arbitrary superpositions of nonclassical harmonic-oscillator states, including squeezed, trisqueezed, and quadsqueezed components, have been demonstrated \cite{Saner2026}. Since continuous-variable tomography can reconstruct finite-Fock amplitudes \cite{LvovskyRaymer2009}, invariants that are computable from reconstructed state families are natural objects to consider.

The Bargmann, or stellar, representation provides the entry point. A pure single-mode state is represented by a holomorphic function \cite{Bargmann1961}. Gaussian pure states correspond to zero-free Bargmann representatives, whereas finite-rank non-Gaussian states carry a finite divisor of stellar zeros. The stellar rank counts these zeros and gives a hierarchy of non-Gaussian states \cite{Chabaud2020}. Recent developments show that this hierarchy is becoming operational: approximate stellar rank constrains approximate Gaussian conversion \cite{Hahn2026}, witness-based approaches connect stellar rank with measurable criteria \cite{Chabaud2021Certification,Provaznik2026Witnesses}, and control parameters have been proposed to go beyond stellar rank in scalable optical generation \cite{Hanamura2025Beyond}. Related invariants such as symplectic rank emphasize the importance of Gaussian-compatible structure in non-Gaussian resource theory \cite{Mele2026Symplectic}.

The present article takes a complementary route. It does not introduce another scalar measure of non-Gaussianity. Instead, it assigns a topological invariant to regular families of finite-rank stellar states. When a state depends continuously on preparation parameters, its stellar zeros move in the complex plane. Along a closed parameter loop, the unordered divisor returns to itself, but the individual zeros can undergo a nontrivial exchange. Stellar rank counts the zeros; braid monodromy records their transport inside the fixed-rank stratum.

The basic mathematical mechanism is simple and rigid. In the finite-Fock chart, a degree-$r$ state is represented by a monic polynomial of degree $r$. The regular locus, where all zeros are simple, is biholomorphic to the unordered configuration space $\Conf_r(\C)/S_r$. Its fundamental group is the Artin braid group $B_r$ \cite{Artin1947,FadellNeuwirth1962,FoxNeuwirth1962,Birman1974,KasselTuraev2008}. A regular parameterized family therefore defines a braid-monodromy homomorphism from the fundamental group of the parameter space to $B_r$. This is classical root monodromy, but here the moving roots are stellar zeros of non-Gaussian bosonic states.

The same idea extends beyond strict finite-Fock polynomial representatives. We also consider Bargmann representatives of the form $E_{\tau,\mu}P$, where $E_{\tau,\mu}$ is an admissible zero-free Gaussian factor and $P$ is a monic polynomial. This separates the Gaussian part, which carries no zeros, from the finite stellar divisor. The viewpoint is consistent with recent stellar-decomposition approaches to Gaussian and non-Gaussian photonic constructions \cite{Motamedi2026StellarDecomposition}, and with geometric uses of finite stellar polynomials in multimode settings \cite{Lopetegui2025Geometric}. In the monomode setting considered here, the Gaussian parameters form a contractible fiber over the same configuration-space base. The braid topology therefore belongs to the stellar divisor, not to a particular finite-Fock coordinate chart.

The scope is deliberately narrow. We treat pure, single-mode states of finite stellar rank. Mixed states, channels, multimode zero varieties, and propagation dynamics are left aside. The sheaf language is also kept minimal: it only records local labelings of stellar zeros and the obstruction to choosing them globally. This restricted setting is already sufficient to make the braid group intrinsic while keeping the construction readable.

Three limitations should be kept in view. First, braid monodromy is not a scalar non-Gaussianity measure and is not attached to an isolated state alone; it is an invariant of a regular loop or family. Second, the trinomial slices studied below are not claimed to realize the full braid group by themselves; they are exactly solvable probes of the full regular stratum. Third, the mathematical mechanism is classical root monodromy, but the roots here are not arbitrary algebraic data: they are the stellar zeros of physically meaningful Bargmann representatives, reconstructible from finite-Fock state data.

\paragraph{Main results and organization.}
The paper proceeds from the physical state space to the topological invariant.
\begin{enumerate}[label=(\roman*),leftmargin=2.4em]
  \item We define the regular degree-$r$ finite-Fock stellar stratum and identify it biholomorphically with $\Conf_r(\C)/S_r$. Consequently, its fundamental group is $B_r$, and the ordered stellar cover recovers the pure-braid exact sequence.
  \item We define stellar braid monodromy for regular parameterized families, prove its homotopy stability inside the regular stratum, and show that every braid is realized in the full regular stratum.
  \item We connect the abstract construction to finite-Fock state engineering. The cubic subspace $\mathrm{span}\{\ket{0},\ket{1},\ket{2},\ket{3}\}$ is used as a minimal experimentally motivated laboratory, while trinomial slices give exactly solvable discriminants and local half-twists.
  \item We include zero-trajectory figures for binomial and trinomial loops, so that the passage from a parameter loop to a stellar braid can be read visually.
  \item We extend the finite-Fock chart to Gaussian-polynomial finite stellar divisors $E_{\tau,\mu}P$ and show that the Gaussian parameters form a contractible fiber over the stellar configuration space.
  \item We explain how the invariant can be reconstructed from finite-Fock amplitudes and give a minimal sheaf interpretation in terms of local labelings of zeros.
\end{enumerate}

\paragraph{Standing conventions.}
The integer $r$ always denotes the stellar degree, namely the degree of the polynomial factor whose zeros form the stellar divisor. Thus $\Hil_r=\mathrm{span}\{\ket{0},\ldots,\ket{r}\}$ has dimension $r+1$. A family is regular when the stellar degree remains equal to $r$ and the stellar divisor is reduced, i.e. all $r$ zeros are simple. In the finite-Fock chart this excludes both the hyperplane $c_r=0$ and the discriminant locus. All braid groups are based at a chosen regular stellar configuration; changing the initial labeling conjugates the corresponding braid representative, so unlabelled loops naturally define conjugacy classes.

\paragraph{Notation guide.}
The paper uses two closely related regular strata. The finite-Fock chart is denoted by $\Sr$ and consists of monic degree-$r$ Bargmann polynomials with simple zeros. The Gaussian-extended finite-divisor stratum is denoted by $\mathcal G_r^{\mathrm{reg}}$ and consists of admissible representatives $E_{\tau,\mu}P$ with $P$ monic of degree $r$ and simple zeros. The map $\Zmap$ records the unordered zero configuration of a finite-Fock polynomial, while $\Pi_\star$ records the same stellar divisor after forgetting the zero-free Gaussian factor. Both maps land in $\Conf_r(\C)/S_r$.
All numerical figures are generated from the roots of the displayed Bargmann polynomials; the accompanying Python script reproduces the figure files included with the manuscript.

\section{Finite-rank stellar states in Bargmann representation}

This section fixes the finite-Fock chart in which the first part of the construction is exact: projective states of stellar degree $r$ are identified with monic polynomials of degree $r$.

Let $\Fock$ be the single-mode bosonic Fock space with number basis $\{\ket{n}\}_{n\geq 0}$. We use the Bargmann convention in which a pure state
\begin{equation}
  \ket{\psi}=\sum_{n\geq 0} c_n \ket{n}
\end{equation}
is represented by the holomorphic function
\begin{equation}
  F_\psi(z)=\sum_{n\geq 0}\frac{c_n}{\sqrt{n!}}z^n.
  \label{eq:bargmann-function}
\end{equation}
Up to a nonvanishing Gaussian factor and convention-dependent conjugations, the zeros of $F_\psi$ coincide with the stellar zeros of the state. Since nonvanishing holomorphic factors do not change zeros, the polynomial part carries the finite stellar divisor.

For a fixed integer $r\geq 1$, define the truncated Fock subspace
\begin{equation}
  \Hil_r=\mathrm{span}\{\ket{0},\ldots,\ket{r}\}.
\end{equation}
A nonzero vector in $\Hil_r$ defines a point in the projective space
\begin{equation}
  \Proj(\Hil_r)\simeq \CP^r.
\end{equation}
We restrict to the affine chart
\begin{equation}
  U_r=\{[\psi]\in \Proj(\Hil_r): c_r\neq 0\},
\end{equation}
so that $F_\psi$ has degree exactly $r$. The restriction to $c_r\neq 0$ avoids zeros at infinity and lets us work with configurations in the affine complex plane $\C$.

Write
\begin{equation}
  F_\psi(z)=\sum_{n=0}^{r} a_n z^n,
  \qquad
  a_n=\frac{c_n}{\sqrt{n!}},
  \qquad
  a_r\neq 0.
\end{equation}
The projective class $[\psi]$ determines the monic polynomial
\begin{equation}
  P_\psi(z)=\frac{1}{a_r}F_\psi(z)
  =z^r+b_{r-1}z^{r-1}+\cdots+b_0,
  \label{eq:monic-polynomial}
\end{equation}
where
\begin{equation}
  b_n=\frac{a_n}{a_r}
  =\frac{c_n/\sqrt{n!}}{c_r/\sqrt{r!}},
  \qquad 0\leq n\leq r-1.
  \label{eq:monic-coefficients}
\end{equation}
The zeros of $P_\psi$ are exactly the finite stellar zeros of $[\psi]$ in the degree-$r$ chart.

\begin{lemma}[Projective monic chart]
\label{lem:projective-monic-chart}
The map
\begin{equation}
  \Theta_r:U_r\longrightarrow \C^r,
  \qquad
  [\psi]\longmapsto (b_0,\ldots,b_{r-1}),
  \label{eq:theta-chart}
\end{equation}
where the $b_n$ are defined by \eqref{eq:monic-coefficients}, is a biholomorphic affine chart. Its inverse sends $(b_0,\ldots,b_{r-1})$ to the projective state whose Bargmann polynomial is
\begin{equation}
  z^r+b_{r-1}z^{r-1}+\cdots+b_0.
\end{equation}
\end{lemma}

\begin{proof}
The ratios $b_n=a_n/a_r$ are unchanged if the representative vector $\ket{\psi}$ is multiplied by a nonzero scalar. Thus $\Theta_r$ is well defined on projective classes. Since $a_r\neq 0$ on $U_r$, every class has a unique representative at the level of Bargmann polynomials for which the leading coefficient is one. This representative is precisely the monic polynomial \eqref{eq:monic-polynomial}. Conversely, any point $(b_0,\ldots,b_{r-1})\in\C^r$ defines the coefficients $a_r=1$, $a_n=b_n$ for $0\leq n\leq r-1$, hence the Fock coefficients $c_n=\sqrt{n!}\,a_n$. These two constructions are inverse. Both are given by holomorphic coordinate functions on the affine chart, proving the claim.
\end{proof}

\begin{definition}[Regular stellar stratum]
The regular stellar stratum of degree $r$ is
\begin{equation}
  \Sr=
  \left\{[\psi]\in U_r:
  F_\psi \text{ has } r \text{ simple zeros in }\C
  \right\}.
  \label{eq:regular-stratum}
\end{equation}
The excluded locus is the stellar discriminant
\begin{equation}
  \Delta_{\star,r}
  =
  \left\{[\psi]\in U_r:
  \Disc(P_\psi)=0
  \right\}.
\end{equation}
Thus
\begin{equation}
  \Sr=U_r\setminus\Delta_{\star,r}.
\end{equation}
\end{definition}

\begin{lemma}[Discriminant hypersurface]
\label{lem:discriminant-hypersurface}
Under the chart $\Theta_r$, the stellar discriminant is the zero locus in $\C^r$ of the classical discriminant polynomial of the monic polynomial \eqref{eq:monic-polynomial}. Therefore $\Sr$ is the complement of an algebraic hypersurface in $U_r\simeq\C^r$.
\end{lemma}

\begin{proof}
For a degree-$r$ polynomial
\begin{equation}
  P(z)=\prod_{j=1}^{r}(z-z_j),
\end{equation}
the discriminant is
\begin{equation}
  \Disc(P)=\prod_{1\leq i<j\leq r}(z_i-z_j)^2.
\end{equation}
It is symmetric in the roots and hence, by the fundamental theorem of symmetric polynomials, can be expressed as a polynomial in the monic coefficients $(b_0,\ldots,b_{r-1})$. It vanishes exactly when two roots coincide. This is precisely the complement of the regularity condition in \eqref{eq:regular-stratum}.
\end{proof}

The distinction between $U_r$ and $\Sr$ is important. The chart $U_r$ fixes the degree. The complement of the discriminant fixes the regularity of the zero configuration. Only on $\Sr$ can the zeros be transported without collisions.

\begin{remark}[Standing regularity convention]
\label{rem:standing-regularity}
Throughout the paper, a regular family of degree-$r$ stellar states means a continuous map into $\Sr$. Thus the leading Bargmann coefficient never vanishes and no two stellar zeros collide along the family. This convention is essential: crossing $c_r=0$ would move a zero to infinity in the affine chart, while crossing $\Delta_{\star,r}$ would leave the configuration space and allow the braid type to change.
\end{remark}

\section{Experimentally motivated finite-Fock families}

The preceding construction starts from monic polynomials, but its physical input is a finite-Fock state. Finite Fock superpositions are standard targets in non-Gaussian state engineering, and they provide the natural bridge between the abstract configuration-space geometry and experimentally accessible state families. We use the full finite-rank stratum for the intrinsic topology, finite-dimensional Fock subspaces for physical motivation, and sparse trinomial slices for exactly solvable discriminants and local braid generators.

The smallest experimentally motivated setting with genuinely nontrivial braid structure is the cubic subspace
\begin{equation}
  \Hil_3=\mathrm{span}\{\ket{0},\ket{1},\ket{2},\ket{3}\}.
\end{equation}
This subspace is natural in view of optical schemes generating arbitrary superpositions of photon-number states up to three photons \cite{Yukawa2013}. More generally, the same finite-Fock language applies to engineered nonclassical oscillator superpositions in which the bosonic mode is controlled through a finite effective basis \cite{Saner2026}.

Consider a state
\begin{equation}
  \ket{\psi}=c_0\ket{0}+c_1\ket{1}+c_2\ket{2}+c_3\ket{3},
  \qquad c_3\neq0.
\end{equation}
Its Bargmann polynomial is
\begin{equation}
  F_\psi(z)=c_0+c_1z+\frac{c_2}{\sqrt{2}}z^2+\frac{c_3}{\sqrt{6}}z^3.
\end{equation}
After projective normalization to a monic cubic, it takes the form
\begin{equation}
  P_\psi(z)=z^3+\alpha z^2+\beta z+\gamma,
  \label{eq:monic-cubic-stellar-polynomial}
\end{equation}
with
\begin{equation}
  \alpha=\frac{c_2/\sqrt{2}}{c_3/\sqrt{6}}
  =\sqrt{3}\,\frac{c_2}{c_3},
  \qquad
  \beta=\frac{c_1}{c_3/\sqrt{6}}
  =\sqrt{6}\,\frac{c_1}{c_3},
  \qquad
  \gamma=\frac{c_0}{c_3/\sqrt{6}}
  =\sqrt{6}\,\frac{c_0}{c_3}.
  \label{eq:cubic-normalized-coefficients}
\end{equation}
Thus the degree-three regular stellar stratum is the complement of the cubic discriminant in the affine space of coefficients $(\alpha,\beta,\gamma)$.

\begin{proposition}[Cubic stellar discriminant]
\label{prop:cubic-stellar-discriminant}
For the normalized cubic stellar polynomial
\begin{equation}
  P(z)=z^3+\alpha z^2+\beta z+\gamma,
\end{equation}
the multiple-root locus is the hypersurface
\begin{equation}
  \Delta_3(\alpha,\beta,\gamma)
  =
  \alpha^2\beta^2
  -4\beta^3
  -4\alpha^3\gamma
  -27\gamma^2
  +18\alpha\beta\gamma
  =0.
  \label{eq:cubic-discriminant}
\end{equation}
Consequently, the regular cubic finite-Fock family is
\begin{equation}
  \mathcal E_3^{\mathrm{reg}}
  =
  \{(\alpha,\beta,\gamma)\in\C^3:
  \Delta_3(\alpha,\beta,\gamma)\neq0\},
\end{equation}
and its fundamental group is naturally the three-strand braid group $B_3$.
\end{proposition}

\begin{proof}
The formula \eqref{eq:cubic-discriminant} is the classical discriminant of a monic cubic. It vanishes exactly when the polynomial has a multiple root. Through the Bargmann normalization \eqref{eq:cubic-normalized-coefficients}, this is precisely the condition that two stellar zeros coalesce. The last statement follows from the general identification of the regular degree-three stellar stratum with $\Conf_3(\C)/S_3$.
\end{proof}

\begin{remark}[Why keep trinomial slices?]
The cubic family \eqref{eq:monic-cubic-stellar-polynomial} is the natural experimental finite-Fock setting. However, its discriminant complement is already three-dimensional and its braid monodromy is not the easiest place to display explicit generators. Trinomial families such as
\begin{equation}
  1+A z^p+B z^q
\end{equation}
should therefore be understood as controlled algebraic slices of the finite-Fock state space. They keep a direct Fock-state interpretation while allowing closed-form discriminants and local half-twist computations.
\end{remark}

\begin{example}[Trinomial slices inside the cubic family]
The elementary cubic slices used later are obtained from \eqref{eq:monic-cubic-stellar-polynomial} by setting one normalized coefficient to zero. For example,
\begin{equation}
  1+A z+B z^3
\end{equation}
becomes, after division by $B$,
\begin{equation}
  z^3+\frac{A}{B}z+\frac{1}{B},
\end{equation}
so it corresponds to the slice $\alpha=0$ in the cubic coefficient space. Similarly,
\begin{equation}
  1+A z^2+B z^3
\end{equation}
corresponds to the slice $\beta=0$. These slices retain the physical interpretation as sparse Fock superpositions, while making the discriminant geometry transparent. The quadratic case $(p,q)=(1,2)$, used below as the simplest two-strand illustration, is the analogous lower-rank slice in $\Hil_2$.
\end{example}

\subsection{From root monodromy to stellar monodromy}
\label{subsec:physical-novelty}

At the technical level, the construction uses the classical monodromy of roots of a polynomial. The physical point, however, is that the polynomial is not an auxiliary algebraic object chosen independently of the state. It is the Bargmann representative of a bosonic quantum state, and its coefficients are fixed by Fock amplitudes through the normalization factors $1/\sqrt{n!}$. Thus a path in coefficient space is, in the present setting, a path of projective finite-Fock states.

This distinction will be important throughout the examples. A loop of preparation amplitudes
\begin{equation}
  [\psi_t]\in\Proj(\Hil_r),
  \qquad 0\leq t\leq 1,
\end{equation}
induces a loop of stellar divisors only after passing through the Bargmann map. The resulting braid is therefore attached to a family of non-Gaussian states, not merely to an abstract family of monic polynomials. Conversely, the monic chart shows that every regular finite stellar divisor can be reconstructed as a finite-Fock state. This is why root monodromy becomes a state-space invariant in the stellar representation.

The sparse families used below should be read in this sense. For example,
\begin{equation}
  1+A z^p+B z^q
\end{equation}
corresponds to the physical superposition
\begin{equation}
  \ket{0}+u\ket{p}+v\ket{q},
  \qquad
  A=\frac{u}{\sqrt{p!}},
  \qquad
  B=\frac{v}{\sqrt{q!}}.
\end{equation}
The braid is generated by varying the amplitudes $u$ and $v$ while remaining away from the stellar discriminant. In this sense, braid monodromy refines the stellar-rank classification by detecting how non-Gaussian stellar features are transported within a fixed finite-rank stratum.

\section{Stellar configurations and braid groups}

Let
\begin{equation}
  \Conf_r(\C)=
  \{(z_1,\ldots,z_r)\in \C^r:
  z_i\neq z_j \text{ for } i\neq j\}
\end{equation}
be the ordered configuration space of $r$ points in the complex plane. The symmetric group $S_r$ acts freely on $\Conf_r(\C)$ by permuting labels. The quotient
\begin{equation}
  \Conf_r(\C)/S_r
\end{equation}
is the unordered configuration space. A point of this quotient is an unordered set $\{z_1,\ldots,z_r\}$ of pairwise distinct complex numbers.

The elementary symmetric functions define the Vieta map
\begin{equation}
  \Vieta:\Conf_r(\C)\longrightarrow \C^r,
  \qquad
  (z_1,\ldots,z_r)\longmapsto (b_0,\ldots,b_{r-1}),
  \label{eq:vieta-map}
\end{equation}
where
\begin{equation}
  \prod_{j=1}^r(z-z_j)
  =z^r+b_{r-1}z^{r-1}+\cdots+b_0.
\end{equation}
Since the coefficients are symmetric functions of the roots, $\Vieta$ is constant on $S_r$-orbits and therefore induces a map
\begin{equation}
  \overline{\Vieta}:\Conf_r(\C)/S_r
  \longrightarrow
  \C^r\setminus\{\Disc=0\}.
  \label{eq:quotient-vieta}
\end{equation}

\begin{lemma}[Vieta identification]
\label{lem:vieta-identification}
The induced Vieta map \eqref{eq:quotient-vieta} is a biholomorphism onto the complement of the discriminant hypersurface.
\end{lemma}

\begin{proof}
The map is well defined because the coefficients of a monic polynomial are symmetric polynomial functions of its roots. It is injective because two unordered configurations with the same elementary symmetric functions define the same monic polynomial, hence the same unordered set of roots. It is surjective onto $\C^r\setminus\{\Disc=0\}$ because any monic polynomial with nonzero discriminant has $r$ distinct roots in $\C$ by the fundamental theorem of algebra.

Continuity of $\overline{\Vieta}$ follows from its polynomial expression in the roots. Its inverse sends a square-free monic polynomial to its unordered set of roots. This inverse is continuous because the roots of a polynomial depend continuously on its coefficients as an unordered set; equivalently, one may use small disjoint contours around the roots and the argument principle to continue each local root cluster under small coefficient perturbations. In local branches where the roots are labeled, the inverse is holomorphic. Therefore the quotient map is a homeomorphism and, with the standard quotient complex structure on $\Conf_r(\C)/S_r$, an isomorphism of complex manifolds.
\end{proof}

\begin{theorem}[Stellar configuration theorem]
\label{thm:stellar-configuration}
The map sending a regular degree-$r$ stellar state to the unordered set of zeros of its Bargmann polynomial defines a natural biholomorphism
\begin{equation}
  \Zmap:\Sr\longrightarrow \Conf_r(\C)/S_r.
  \label{eq:stellar-configuration-map}
\end{equation}
More precisely, in the monic chart $\Theta_r$, this map is the inverse of the Vieta identification.
\end{theorem}

\begin{proof}
By \cref{lem:projective-monic-chart}, a point of $U_r$ is equivalently a monic degree-$r$ Bargmann polynomial. By \cref{lem:discriminant-hypersurface}, the regular stratum $\Sr$ corresponds exactly to the square-free monic polynomials, i.e. to $\C^r\setminus\{\Disc=0\}$ in coefficient coordinates. By \cref{lem:vieta-identification}, this space of square-free monic polynomials is naturally homeomorphic to $\Conf_r(\C)/S_r$ by taking roots. Composing these biholomorphic identifications gives \eqref{eq:stellar-configuration-map}.
\end{proof}

\begin{corollary}[Intrinsic braid group]
\label{cor:braid-group}
The fundamental group of the regular stellar stratum is the Artin braid group:
\begin{equation}
  \pi_1(\Sr)\simeq B_r.
  \label{eq:fundamental-group-braid}
\end{equation}
\end{corollary}

\begin{proof}
By \cref{thm:stellar-configuration}, $\Sr$ is homeomorphic, indeed biholomorphic, to $\Conf_r(\C)/S_r$. The fundamental group of the unordered configuration space of $r$ points in the complex plane is the Artin braid group $B_r$ \cite{Artin1947,Birman1974,KasselTuraev2008}. Hence $\pi_1(\Sr)\simeq B_r$.
\end{proof}

\begin{remark}[Physical meaning]
\Cref{cor:braid-group} is elementary from the point of view of polynomial roots, but it is structural for the stellar representation. It says that the braid group is not an external decoration added to the classification of non-Gaussian states. It is the fundamental group of the regular finite-rank stellar stratum itself. Stellar rank selects the number of zeros; braid topology describes the possible nontrivial transports of these zeros at fixed rank.
\end{remark}

\subsection{The ordered stellar cover and pure braids}

The unordered configuration space is the natural space of zeros of a projective state, because a state does not come with labeled stellar zeros. Nevertheless, many constructions become clearer after introducing the ordered cover.

\begin{definition}[Ordered stellar cover]
The ordered regular stellar cover is
\begin{equation}
  \widetilde{\mathcal S}_r^{\mathrm{reg}}
  =
  \left\{
  ([\psi],z_1,\ldots,z_r)\in \Sr\times\Conf_r(\C):
  P_\psi(z_j)=0\text{ for all }j
  \right\}.
  \label{eq:ordered-stellar-cover}
\end{equation}
The projection
\begin{equation}
  \pi_\star:\widetilde{\mathcal S}_r^{\mathrm{reg}}
  \longrightarrow \Sr,
  \qquad
  ([\psi],z_1,\ldots,z_r)\longmapsto [\psi]
  \label{eq:ordered-stellar-projection}
\end{equation}
forgets the labels of the zeros.
\end{definition}

\begin{proposition}[Ordered cover]
\label{prop:ordered-cover}
The ordered regular stellar cover is naturally isomorphic to the ordered configuration space:
\begin{equation}
  \widetilde{\mathcal S}_r^{\mathrm{reg}}
  \simeq
  \Conf_r(\C).
\end{equation}
Under this identification, the projection \eqref{eq:ordered-stellar-projection} is the quotient map
\begin{equation}
  \Conf_r(\C)\longrightarrow \Conf_r(\C)/S_r.
\end{equation}
It is a regular covering with deck group $S_r$.
\end{proposition}

\begin{proof}
Given an ordered configuration $(z_1,\ldots,z_r)\in\Conf_r(\C)$, form the monic polynomial
\begin{equation}
  P(z)=\prod_{j=1}^{r}(z-z_j).
\end{equation}
By the projective monic chart, this polynomial determines a unique point $[\psi]\in\Sr$. This gives a map
\begin{equation}
  \Conf_r(\C)\longrightarrow \widetilde{\mathcal S}_r^{\mathrm{reg}},
  \qquad
  (z_1,\ldots,z_r)\longmapsto ([\psi],z_1,\ldots,z_r).
\end{equation}
The inverse simply forgets the state and keeps the ordered zero configuration. These maps are inverse to each other. They are continuous, and locally holomorphic after choosing local root branches. The action of $S_r$ permutes the ordered zeros, is free, and leaves the underlying projective state unchanged. Therefore $\pi_\star$ is identified with the standard quotient covering of ordered by unordered configurations.
\end{proof}

\begin{corollary}[Pure-braid exact sequence]
\label{cor:pure-braid-exact-sequence}
The ordered stellar cover yields the standard exact sequence
\begin{equation}
  1\longrightarrow P_r\longrightarrow B_r\longrightarrow S_r\longrightarrow 1,
  \label{eq:pure-braid-exact-sequence}
\end{equation}
where
\begin{equation}
  P_r\simeq \pi_1(\widetilde{\mathcal S}_r^{\mathrm{reg}})
\end{equation}
is the pure braid group.
\end{corollary}

\begin{proof}
By \cref{prop:ordered-cover}, $\widetilde{\mathcal S}_r^{\mathrm{reg}}\simeq\Conf_r(\C)$, hence its fundamental group is the pure braid group $P_r$. The covering has deck group $S_r$, and the quotient has fundamental group $B_r$. The associated exact sequence of a regular covering gives \eqref{eq:pure-braid-exact-sequence}. Concretely, the map $B_r\to S_r$ records only the final permutation of the labeled zeros, while its kernel consists of braids in which each zero returns to its initial label.
\end{proof}

\begin{remark}[Why this cover matters]
The full braid contains more information than the final permutation of zeros. The covering monodromy records the map $B_r\to S_r$, whereas the isotopy class of the strands records the element of $B_r$ itself. This distinction will be important when comparing families with the same permutation monodromy but different pure-braid components.
\end{remark}

\section{Stellar braid monodromy of parameterized families}

Let $\mathcal M$ be a connected parameter space and let
\begin{equation}
  \Phi:\mathcal M\longrightarrow \Sr
\end{equation}
be a continuous family of regular finite-rank stellar states. The stellar configuration map gives
\begin{equation}
  \Zmap\circ\Phi:
  \mathcal M\longrightarrow \Conf_r(\C)/S_r.
\end{equation}
Thus a path in parameter space transports an unordered configuration of stellar zeros.

\begin{definition}[Stellar braid of a loop]
Let $\gamma:[0,1]\to\mathcal M$ be a loop based at $m_0$. The associated stellar braid is the braid represented by the loop
\begin{equation}
  (\Zmap\circ\Phi\circ\gamma):[0,1]
  \longrightarrow \Conf_r(\C)/S_r.
\end{equation}
Equivalently, after choosing an initial labeling of the zeros over $m_0$, one lifts the motion locally to ordered configurations and records the isotopy class of the resulting $r$ strands in $\C\times[0,1]$.
\end{definition}

\begin{proposition}[Path lifting and stellar strands]
\label{prop:path-lifting-strands}
Let $\eta:[0,1]\to\mathcal M$ be a path and choose an ordering of the zeros of $\Phi(\eta(0))$. Then there exists a unique lift of the zero motion to the ordered stellar cover,
\begin{equation}
  (z_1(t),\ldots,z_r(t))\in\Conf_r(\C),
  \qquad 0\leq t\leq 1,
\end{equation}
starting from the chosen ordering. The curves
\begin{equation}
  t\longmapsto (z_j(t),t)\in\C\times[0,1]
\end{equation}
form $r$ disjoint strands. If $\eta$ is a loop, the isotopy class of these strands is the stellar braid of the loop.
\end{proposition}

\begin{proof}
The ordered stellar cover is a covering of the regular stratum by \cref{prop:ordered-cover}. The path-lifting property of coverings gives a unique lift once an initial point in the fiber, i.e. an initial labeling of the zeros, is chosen. Since the path stays inside the regular stratum, the lifted roots remain pairwise distinct for every $t$. Therefore the graphs of the lifted roots in $\C\times[0,1]$ are disjoint strands. For a loop in the unordered configuration space, the endpoint of the lifted ordered configuration may differ from the starting one by a permutation; the resulting strand diagram is exactly the usual geometric representative of a braid.
\end{proof}

\begin{definition}[Stellar braid monodromy]
Choose a base point $m_0\in\mathcal M$. The stellar braid monodromy of the based family $(\Phi,m_0)$ is the induced homomorphism
\begin{equation}
  \rho_\Phi
  =
  (\Zmap\circ\Phi)_*:
  \pi_1(\mathcal M,m_0)
  \longrightarrow
  \pi_1(\Conf_r(\C)/S_r,\Zmap(\Phi(m_0)))
  \simeq B_r.
  \label{eq:stellar-monodromy}
\end{equation}
\end{definition}

\begin{proposition}[Well-defined monodromy]
\label{prop:well-defined-monodromy}
The map $\rho_\Phi$ is a group homomorphism. It is invariant under based homotopy of loops in $\mathcal M$. If the initial labeling of the zeros or the base configuration is changed, the corresponding explicit braid words are conjugated. Hence an unbased closed family naturally determines a conjugacy class in $B_r$.
\end{proposition}

\begin{proof}
The first statement follows from functoriality of the fundamental group: the continuous map $\Zmap\circ\Phi$ induces a group homomorphism on based fundamental groups. Based homotopic loops have the same image in $\pi_1(\Conf_r(\C)/S_r)$, proving homotopy invariance.

To obtain explicit braid words, one chooses a base configuration and an ordering of its $r$ points. A different ordering acts by changing the identification of the fiber with a labeled configuration, which conjugates the braid representative. Similarly, changing the base point along a path changes the induced fundamental-group isomorphism by conjugation. Therefore the based monodromy is canonical after choices, while the unbased invariant is naturally a conjugacy class.
\end{proof}

\begin{proposition}[Stability under regular homotopy]
\label{prop:regular-homotopy-stability}
Let $\gamma_0$ and $\gamma_1$ be two loops in $\mathcal M$ based at $m_0$. If they are homotopic through loops whose images under $\Phi$ remain in $\Sr$, then they have the same stellar braid monodromy. Conversely, the braid associated with a loop can change only if a deformation crosses the stellar discriminant or leaves the degree-$r$ chart.
\end{proposition}

\begin{proof}
The first statement is the usual homotopy invariance of the induced map on fundamental groups applied to $\Zmap\circ\Phi$. If a deformation crosses $\Delta_r$, the zero configuration ceases to be a point of $\Conf_r(\C)/S_r$ at the crossing because at least two zeros coincide. If the leading coefficient vanishes, the affine degree-$r$ chart is left. These are precisely the two ways in which the loop can fail to remain in the regular configuration space used to define the braid.
\end{proof}

\begin{proposition}[Realization of braids]
\label{prop:realization}
Every braid $\beta\in B_r$ is realized as the stellar braid monodromy of some loop in the full regular stratum $\Sr$.
\end{proposition}

\begin{proof}
Represent $\beta$ by a geometric braid, equivalently by a loop
\begin{equation}
  C(t)=\{z_1(t),\ldots,z_r(t)\}
  \in \Conf_r(\C)/S_r,
  \qquad C(0)=C(1),
\end{equation}
whose points remain pairwise distinct for all $t$. For each $t$, define the monic polynomial
\begin{equation}
  P_t(z)=\prod_{j=1}^{r}(z-z_j(t)).
\end{equation}
Although the labeling used to write the product is auxiliary, the polynomial is symmetric in the roots and therefore depends only on the unordered configuration $C(t)$. Its coefficients vary continuously, and smoothly whenever the braid is smooth. By the projective monic chart, $P_t$ determines a unique projective state $[\psi_t]\in\Sr$. The loop $t\mapsto[\psi_t]$ has stellar zero motion $C(t)$, hence its stellar braid is precisely $\beta$.
\end{proof}

\begin{remark}
\Cref{prop:realization} is a statement about the full regular stratum. A restricted physical subfamily, such as a two-parameter trinomial family, need not realize all braids in $B_r$. Its image may probe only a distinguished subgroup or a distinguished set of conjugacy classes.
\end{remark}

\section{Stellar discriminant and local braid generators}

The stellar discriminant is the locus where the regular transport of zeros can fail. In the degree-$r$ chart it is
\begin{align}
  \Delta_{\star,r}
  & =
  \{[\psi]\in U_r: F_\psi \text{ has at least one multiple zero}\} \\
  & =
  \{[\psi]\in U_r: \Disc(F_\psi)=0\}.
\end{align}
The regular stratum is the complement $U_r\setminus\Delta_{\star,r}$.

\begin{lemma}[Local square-root normal form]
\label{lem:local-square-root-normal-form}
Let $P_0$ be a monic degree-$r$ polynomial with exactly one double root $z_0$ and all other roots simple. Near $P_0$, the discriminant is locally a smooth hypersurface away from higher multiplicity strata, and a one-parameter deformation transverse to it has the local two-root normal form
\begin{equation}
  w^2-\lambda=0,
  \label{eq:local-square-root-model}
\end{equation}
up to multiplication by a nonvanishing analytic factor and analytic changes of coordinates.
\end{lemma}

\begin{proof}
Translate the double root to the origin. In a small disc containing no other roots, the Weierstrass preparation theorem writes the local factor of a nearby polynomial as
\begin{equation}
  w^2+a_1 w+a_0,
\end{equation}
where $a_0,a_1$ are analytic functions of the polynomial coefficients and vanish at $P_0$. Completing the square gives
\begin{equation}
  \left(w+\frac{a_1}{2}\right)^2-
  \left(\frac{a_1^2}{4}-a_0\right)=0.
\end{equation}
The local discriminant is therefore given by
\begin{equation}
  \lambda:=\frac{a_1^2}{4}-a_0=0.
\end{equation}
Away from higher multiplicity points, this is a transverse coordinate to the discriminant. The remaining roots are simple and lie outside the chosen disc, so they vary analytically and do not affect the local two-root model.
\end{proof}

\begin{proposition}[Local half-twist around a simple collision]
\label{prop:local-half-twist}
Near a generic point of $\Delta_{\star,r}$ where exactly two zeros coalesce and all other zeros remain simple, a small loop transverse to the discriminant produces a local half-twist of the two colliding stellar zeros. Hence its braid monodromy is a conjugate of $\sigma_i$ or $\sigma_i^{-1}$ in $B_r$, depending on orientation and labeling conventions.
\end{proposition}

\begin{proof}
By \cref{lem:local-square-root-normal-form}, the two colliding roots are locally modeled by
\begin{equation}
  w_\pm(\lambda)=\pm\sqrt{\lambda}.
\end{equation}
When $\lambda$ makes one turn around the origin, the two branches are exchanged. Their graphs in the local disc times the interval form a half-twist. The remaining roots stay simple and disjoint from the local disc, so they contribute only vertical strands up to isotopy. Embedding this local two-strand motion among the $r$ strands gives a conjugate of an elementary Artin generator, with inverse corresponding to the opposite orientation.
\end{proof}

\begin{remark}
This result gives the local geometric meaning of the discriminant: crossing the discriminant corresponds to a collision of stellar zeros, while encircling a smooth point of it produces a braid generator. Thus braid monodromy is controlled by the topology of the complement of the stellar discriminant.
\end{remark}

\begin{remark}[Regular and singular parts of the stellar discriminant]
\label{rem:discriminant-strata}
The preceding proposition uses only the generic part of the stellar discriminant. More precisely, let
\begin{equation}
  \Delta_{\star,r}^{\mathrm{gen}}
  \subset \Delta_{\star,r}
\end{equation}
denote the locus where the polynomial has exactly one double zero and all other zeros are simple. This is the smooth codimension-one part relevant for local half-twists. The complement
\begin{equation}
  \Delta_{\star,r}\setminus \Delta_{\star,r}^{\mathrm{gen}}
\end{equation}
contains higher degeneracies, such as triple zeros or several simultaneous double zeros. These higher strata can produce more complicated local monodromy, but they are not needed for the finite-rank construction developed here. Throughout the explicit examples, small loops are taken around smooth points of the discriminant unless otherwise stated.
\end{remark}

\section{Trinomial Fock families and explicit discriminants}

We now introduce a simple class of finite-rank families that is both physically meaningful and algebraically nontrivial. For integers $1\leq p<q$, consider
\begin{equation}
  \ket{\psi_{u,v}^{(p,q)}}
  =
  \ket{0}+u\ket{p}+v\ket{q}.
  \label{eq:trinomial-state}
\end{equation}
Its Bargmann polynomial is
\begin{equation}
  F_{u,v}^{(p,q)}(z)
  =
  1+A z^p+B z^q,
  \qquad
  A=\frac{u}{\sqrt{p!}},
  \qquad
  B=\frac{v}{\sqrt{q!}}.
  \label{eq:trinomial-polynomial}
\end{equation}
We assume $B\neq 0$, so that the degree is exactly $q$. This family is a low-dimensional physical slice of the full degree-$q$ stellar stratum. Its regular parameter space will be denoted by
\begin{equation}
  \mathcal T_{p,q}^{\mathrm{reg}}
  =
  \left\{(A,B)\in\C\times\C^*:
  1+A z^p+B z^q\text{ has }q\text{ simple roots}
  \right\}.
  \label{eq:trinomial-regular-parameter-space}
\end{equation}

\begin{proposition}[Discriminant of trinomial Fock families]
\label{prop:trinomial-discriminant}
Let
\begin{equation}
  F(z)=1+A z^p+B z^q,
  \qquad
  1\leq p<q,
  \qquad
  B\neq 0.
  \label{eq:trinomial-polynomial-general}
\end{equation}
Set
\begin{equation}
  d=\gcd(p,q),
  \qquad
  p=d\bar p,
  \qquad
  q=d\bar q,
  \qquad
  \gcd(\bar p,\bar q)=1.
\end{equation}
Then the multiple-root locus in the parameter space $\C\times\C^*$ is the reduced plane curve
\begin{equation}
  A^{\bar q}
  =
  (-1)^{\bar q}
  \frac{\bar q^{\bar q}}
  {\bar p^{\bar p}(\bar q-\bar p)^{\bar q-\bar p}}
  B^{\bar p}.
  \label{eq:trinomial-discriminant-AB-reduced}
\end{equation}
In particular, when $p$ and $q$ are coprime,
\begin{equation}
  A^q
  =
  (-1)^q\frac{q^q}{p^p(q-p)^{q-p}}B^p.
  \label{eq:trinomial-discriminant-AB}
\end{equation}
Equivalently, in the physical coefficients $u,v$, one substitutes
\begin{equation}
  A=\frac{u}{\sqrt{p!}},
  \qquad
  B=\frac{v}{\sqrt{q!}}.
\end{equation}
\end{proposition}

\begin{proof}
First suppose $\gcd(p,q)=1$. A multiple root $z$ satisfies
\begin{equation}
  F(z)=0,
  \qquad
  F'(z)=0.
\end{equation}
Since $F(0)=1$, such a root cannot occur at $z=0$. The derivative equation gives
\begin{equation}
  pA z^{p-1}+qB z^{q-1}=0,
\end{equation}
or, after multiplication by $z$,
\begin{equation}
  pA z^p+qB z^q=0.
  \label{eq:derivative-relation-v4}
\end{equation}
Solving the linear system formed by \eqref{eq:derivative-relation-v4} and
\begin{equation}
  1+A z^p+B z^q=0
\end{equation}
for $Az^p$ and $Bz^q$ gives
\begin{equation}
  A z^p=-\frac{q}{q-p},
  \qquad
  B z^q=\frac{p}{q-p}.
  \label{eq:trinomial-root-relations-v4}
\end{equation}
Raising the first relation to the power $q$, the second to the power $p$, and dividing eliminates $z^{pq}$:
\begin{equation}
  \frac{A^q}{B^p}
  =
  (-1)^q\frac{q^q}{p^p(q-p)^{q-p}}.
\end{equation}
This proves \eqref{eq:trinomial-discriminant-AB} in the coprime case. Conversely, if this relation holds, choose $z\neq0$ satisfying
\begin{equation}
  z^q=\frac{p}{(q-p)B}.
\end{equation}
Because $p$ is invertible modulo $q$, the coprime relation ensures that one of the $q$ choices of $z$ also satisfies the first equation in \eqref{eq:trinomial-root-relations-v4}. Hence the condition is also sufficient.

If $d=\gcd(p,q)>1$, write $p=d\bar p$ and $q=d\bar q$, and set $y=z^d$. Then
\begin{equation}
  F(z)=G(y),
  \qquad
  G(y)=1+A y^{\bar p}+B y^{\bar q}.
\end{equation}
Since any multiple root of $F$ is nonzero, $F'(z)=d z^{d-1}G'(z^d)$ shows that multiple roots of $F$ are exactly the lifts, under $y=z^d$, of multiple roots of $G$. Applying the coprime result to $G$ gives \eqref{eq:trinomial-discriminant-AB-reduced}.
\end{proof}

\begin{corollary}[Regular trinomial parameter space]
\label{cor:regular-trinomial-parameter-space}
With the notation of \cref{prop:trinomial-discriminant}, the regular parameter space is
\begin{equation}
  \mathcal T_{p,q}^{\mathrm{reg}}
  =
  \left\{(A,B)\in\C\times\C^*:
  A^{\bar q}
  \neq
  (-1)^{\bar q}
  \frac{\bar q^{\bar q}}
  {\bar p^{\bar p}(\bar q-\bar p)^{\bar q-\bar p}}
  B^{\bar p}\right\}.
  \label{eq:regular-trinomial-parameter-space-explicit}
\end{equation}
On this complement, the trinomial family defines a braid monodromy representation
\begin{equation}
  \rho_{p,q}:\pi_1(\mathcal T_{p,q}^{\mathrm{reg}})
  \longrightarrow B_q.
  \label{eq:trinomial-monodromy-representation}
\end{equation}
\end{corollary}

\begin{proof}
The first statement is the vanishing condition of \cref{prop:trinomial-discriminant}, together with the degree condition $B\neq0$. On this complement, the trinomial polynomial has degree $q$ and $q$ simple roots, so it defines a map to $\mathcal S_q^{\mathrm{reg}}$. The braid monodromy representation is then the one induced by the general construction.
\end{proof}

\begin{remark}[Scope of trinomial monodromy]
\label{rem:trinomial-monodromy-scope}
The representation \eqref{eq:trinomial-monodromy-representation} is not classified in this article. In particular, we do not claim that a fixed trinomial slice realizes all of $B_q$. The role of these slices is more modest and more explicit: their discriminants are computable by hand, and loops around smooth discriminant points produce local half-twists by \cref{prop:local-half-twist}. The full realization of arbitrary braids belongs to the complete regular stratum of \cref{prop:realization}, not to every lower-dimensional physical slice.
\end{remark}

\begin{corollary}[Parametric form of the coprime trinomial discriminant]
\label{cor:trinomial-parametric-discriminant}
Assume $\gcd(p,q)=1$. Away from $A=0$ and $B=0$, the discriminant of the family $1+A z^p+B z^q$ is the image of
\begin{equation}
  z\in\C^*
  \longmapsto
  (A(z),B(z))
  =
  \left(
  -\frac{q}{q-p}z^{-p},
  \frac{p}{q-p}z^{-q}
  \right).
  \label{eq:trinomial-parametric-discriminant}
\end{equation}
Equivalently, it is the plane curve \eqref{eq:trinomial-discriminant-AB} in $\C\times\C^*$.
\end{corollary}

\begin{proof}
The parametrization is exactly \eqref{eq:trinomial-root-relations-v4} solved for $A$ and $B$. Eliminating $z$ gives \eqref{eq:trinomial-discriminant-AB}. Conversely, because $p$ and $q$ are coprime, the elimination condition is sufficient for the existence of a nonzero $z$ satisfying both equations in \eqref{eq:trinomial-root-relations-v4}.
\end{proof}

\begin{proposition}[Fixed-$B$ slices]
\label{prop:fixed-B-slices}
Assume $\gcd(p,q)=1$ and fix $B_0\neq0$. The one-parameter slice
\begin{equation}
  A\longmapsto1+A z^p+B_0z^q
\end{equation}
has exactly $q$ discriminant points in the $A$-plane,
\begin{equation}
  A^q
  =
  (-1)^q\frac{q^q}{p^p(q-p)^{q-p}}B_0^p.
\end{equation}
Each such point corresponds to a unique double stellar zero. A small positively oriented loop around one of these points gives a conjugate of an elementary half-twist in $B_q$.
\end{proposition}

\begin{proof}
The discriminant equation is \eqref{eq:trinomial-discriminant-AB} with $B=B_0$, and hence has $q$ distinct solutions for $A$. For each solution, the corresponding double root is determined by the equations
\begin{equation}
  A z^p=-\frac{q}{q-p},
  \qquad
  B_0z^q=\frac{p}{q-p}.
\end{equation}
The coprimality of $p$ and $q$ makes the assignment between the $q$ roots of the second equation and the $q$ discriminant values in the $A$-plane one-to-one. At each such point, the polynomial has exactly one double root and the other roots are simple. Therefore \cref{prop:local-half-twist} applies and gives a conjugate of a local Artin generator.
\end{proof}

\begin{remark}[Non-coprime exponents]
If $d=\gcd(p,q)>1$, the substitution $y=z^d$ shows that the discriminant is governed by the reduced pair $(\bar p,\bar q)$. A loop around a reduced discriminant point may lift to several simultaneous local exchanges among the $z$-roots. To keep the first examples focused on elementary braid generators, we use coprime pairs in the explicit test cases below.
\end{remark}
\begin{example}[Binomial cyclic braid]
\label{ex:binomial-cyclic-braid}
The binomial family
\begin{equation}
  F_\lambda(z)=1+\lambda z^r,
  \qquad \lambda\in\C^*,
\end{equation}
has roots
\begin{equation}
  z_k(\lambda)=(-\lambda^{-1})^{1/r}\exp\left(\frac{2\pi i k}{r}\right),
  \qquad k=0,\ldots,r-1,
\end{equation}
after a local choice of the $r$th root. When $\lambda$ winds once around the origin, the roots undergo a cyclic relabeling. The associated braid is the standard cyclic braid whose permutation is $(1\,2\,\cdots\,r)$. This example is highly symmetric, so it is pedagogically useful but not sufficient to probe the full algebraic richness of $B_r$.
\end{example}

\subsection{Elementary trinomial test cases}

The following cases will be useful later for numerical and geometric illustrations. The first one realizes the elementary two-strand exchange, while the cubic cases provide the smallest setting where braid words can carry information beyond a final permutation.

\begin{example}[Quadratic family]
For $(p,q)=(1,2)$,
\begin{equation}
  F(z)=1+Az+Bz^2.
\end{equation}
The discriminant condition is
\begin{equation}
  A^2=4B,
\end{equation}
which is the usual quadratic discriminant condition for a double root.
\end{example}

\begin{example}[Cubic family]
For $(p,q)=(1,3)$,
\begin{equation}
  F(z)=1+Az+Bz^3.
\end{equation}
The discriminant condition is
\begin{equation}
  A^3=-\frac{27}{4}B.
\end{equation}
Loops around this discriminant curve generate nontrivial braids of the three stellar zeros.
\end{example}

\begin{example}[Asymmetric cubic family]
For $(p,q)=(2,3)$,
\begin{equation}
  F(z)=1+Az^2+Bz^3.
\end{equation}
The discriminant condition is
\begin{equation}
  A^3=-\frac{27}{4}B^2.
\end{equation}
This family is a minimal test case for braid monodromy beyond the purely binomial situation.
\end{example}

\begin{remark}
The trinomial discriminant \eqref{eq:trinomial-discriminant-AB} has the form of a toric plane curve. Thus even very sparse Fock superpositions can produce nontrivial algebraic discriminants. This is the first indication that stellar braid monodromy is not merely the monodromy of arbitrary polynomials, but can probe structured subfamilies selected by photonic state engineering.
\end{remark}

\section{Concrete stellar braids and zero-trajectory examples}
\label{sec:concrete-stellar-braids}

The previous sections define stellar braid monodromy intrinsically. We now record a few elementary loops that will serve as concrete test cases. The goal is not to exhaust the monodromy of the corresponding parameter spaces, but to make the braid construction visible at the level of stellar-zero trajectories.

\subsection{Binomial loops}

\begin{example}[Two-strand binomial exchange]
\label{ex:two-strand-binomial-exchange}
Consider
\begin{equation}
  F_t(z)=1+e^{2\pi i t}z^2,
  \qquad 0\leq t\leq 1.
  \label{eq:quadratic-binomial-loop}
\end{equation}
This is the Bargmann polynomial of a loop of states proportional to
\begin{equation}
  \ket{0}+\sqrt{2}\,e^{2\pi i t}\ket{2}.
\end{equation}
The two zeros satisfy
\begin{equation}
  z^2=-e^{-2\pi i t}.
\end{equation}
A continuous choice of square root gives two branches which are exchanged after one turn in $t$. Hence the loop realizes the generator of $B_2$, up to the sign convention for positive crossings. In symbols,
\begin{equation}
  \rho([F_t])=\sigma_1^{\pm1}\in B_2.
\end{equation}
The sign depends on the orientation convention for the braid projection.
\end{example}

\begin{figure}[htbp]
  \centering
  \includegraphics[width=0.96\linewidth]{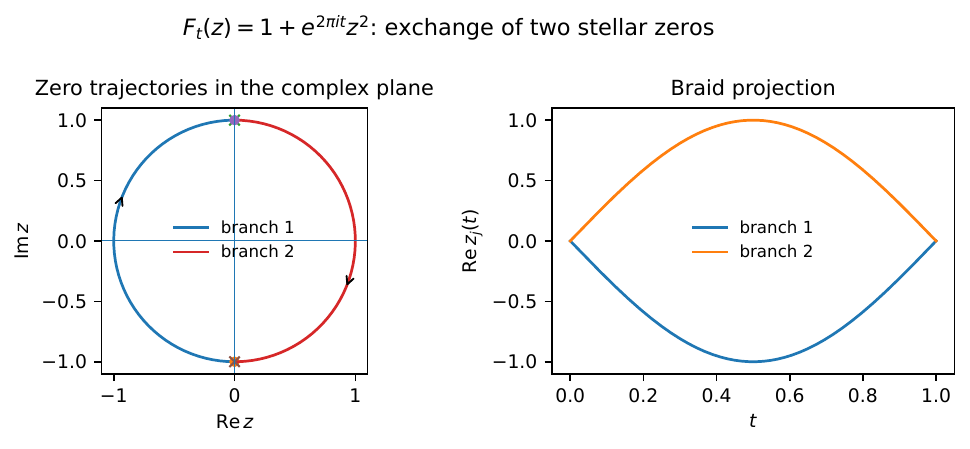}
  \caption{Two-strand stellar braid generated by the binomial family \eqref{eq:quadratic-binomial-loop}. The left panel follows the two stellar zeros in the complex plane as the phase of the finite-Fock state is wound once; the right panel shows a braid projection using the real parts of the transported zeros. The figure illustrates the basic mechanism of the paper: a loop of non-Gaussian states of fixed stellar degree can exchange its stellar zeros and thereby define an element of $B_2$.}
  \label{fig:quadratic-binomial-braid}
\end{figure}

\begin{example}[Cyclic three-strand binomial braid]
\label{ex:cyclic-three-strand-binomial}
Consider
\begin{equation}
  F_t(z)=1+e^{2\pi i t}z^3,
  \qquad 0\leq t\leq 1.
  \label{eq:cubic-binomial-loop}
\end{equation}
This corresponds to a loop of states proportional to
\begin{equation}
  \ket{0}+\sqrt{6}\,e^{2\pi i t}\ket{3}.
\end{equation}
The three zeros form a rotating equilateral configuration. After one turn of the parameter, the unordered set returns to itself but the lifted labeling is cyclically permuted. The associated braid is the standard cyclic braid
\begin{equation}
  \delta_3=\sigma_1\sigma_2
\end{equation}
or its inverse, depending on orientation and labeling conventions. Its image in $S_3$ is a three-cycle.
\end{example}

\begin{figure}[htbp]
  \centering
  \includegraphics[width=0.96\linewidth]{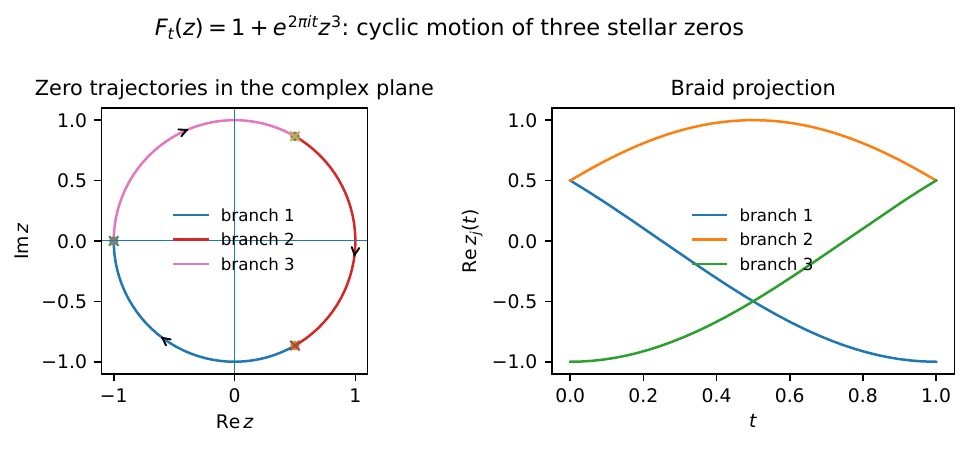}
  \caption{Cyclic three-strand stellar braid for the cubic binomial loop \eqref{eq:cubic-binomial-loop}. The zeros form an equilateral configuration and return to the same unordered divisor after a cyclic relabeling. This is the simplest visual example in which the cubic finite-Fock family realizes a nontrivial braid in $B_3$.}
  \label{fig:cubic-binomial-braid}
\end{figure}

The binomial examples are useful because they are explicit and symmetric. However, they probe only cyclic motions. The next examples show how local half-twists arise from loops around discriminant points in trinomial slices.

\subsection{Loops around trinomial discriminants}

\begin{example}[Quadratic discriminant loop]
\label{ex:quadratic-discriminant-loop}
Fix $B=1$ in the quadratic trinomial family
\begin{equation}
  F_A(z)=1+Az+z^2.
\end{equation}
The discriminant condition is
\begin{equation}
  A^2=4,
\end{equation}
so the slice has two discriminant points, $A=2$ and $A=-2$. At $A=2$, one has
\begin{equation}
  F_2(z)=(z+1)^2,
\end{equation}
so the two stellar zeros collide at $z=-1$. A small loop
\begin{equation}
  A(t)=2+\varepsilon e^{2\pi i t},
  \qquad 0<\varepsilon\ll1,
  \label{eq:quadratic-discriminant-loop}
\end{equation}
encircles a smooth point of the discriminant and gives a local half-twist of the two zeros. Therefore the braid is a conjugate of $\sigma_1^{\pm1}$; since $B_2$ is generated by $\sigma_1$, this is simply the elementary two-strand braid up to orientation.
\end{example}

\begin{figure}[htbp]
  \centering
  \includegraphics[width=0.96\linewidth]{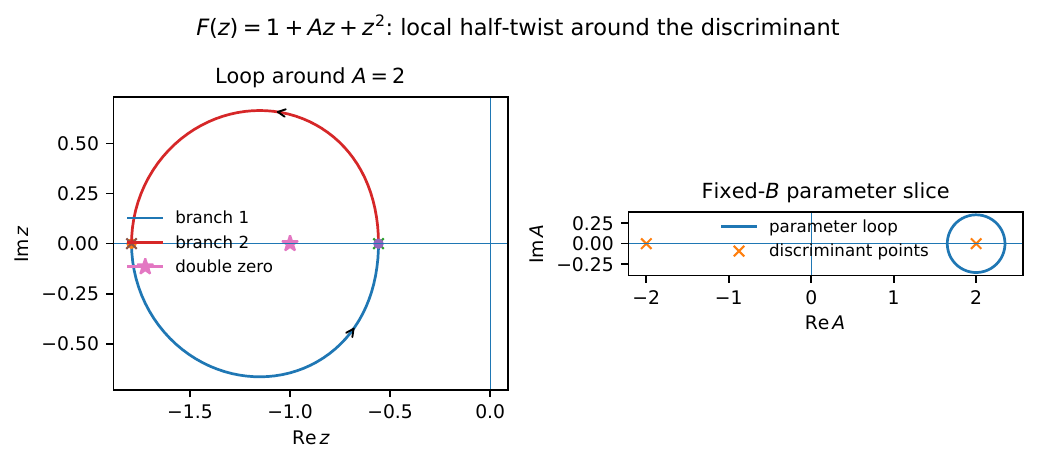}
  \caption{Quadratic discriminant loop and local half-twist. For $F_A(z)=1+Az+z^2$, the fixed-$B$ slice has discriminant points $A=\pm2$. A small loop around $A=2$ produces a collision-avoiding exchange of the two stellar zeros near the double root. This is the local model behind the statement that a smooth point of the stellar discriminant produces an elementary half-twist.}
  \label{fig:quadratic-discriminant-loop}
\end{figure}

\begin{example}[Asymmetric cubic discriminant loop]
\label{ex:asymmetric-cubic-discriminant-loop}
Consider the asymmetric cubic slice
\begin{equation}
  F_A(z)=1+Az^2+z^3,
  \label{eq:asymmetric-cubic-slice}
\end{equation}
which is the case $(p,q)=(2,3)$ with $B=1$. By \cref{prop:trinomial-discriminant}, the discriminant condition is
\begin{equation}
  A^3=-\frac{27}{4}.
\end{equation}
Choose the real discriminant point
\begin{equation}
  A_c=-\frac{3}{2^{2/3}}.
\end{equation}
The corresponding double stellar zero is
\begin{equation}
  z_c=2^{1/3},
\end{equation}
which follows from the equations
\begin{equation}
  A_c z_c^2=-3,
  \qquad
  z_c^3=2.
\end{equation}
A small loop
\begin{equation}
  A(t)=A_c+\varepsilon e^{2\pi i t}
  \label{eq:asymmetric-cubic-loop}
\end{equation}
produces a half-twist of the two roots colliding at $z_c$, while the third root remains simple. In a labeling adapted to the colliding pair, the local braid is $\sigma_i^{\pm1}$; in a global labeling fixed at the base point, it appears as a conjugate of such a generator in $B_3$.
\end{example}

\begin{figure}[htbp]
  \centering
  \includegraphics[width=0.96\linewidth]{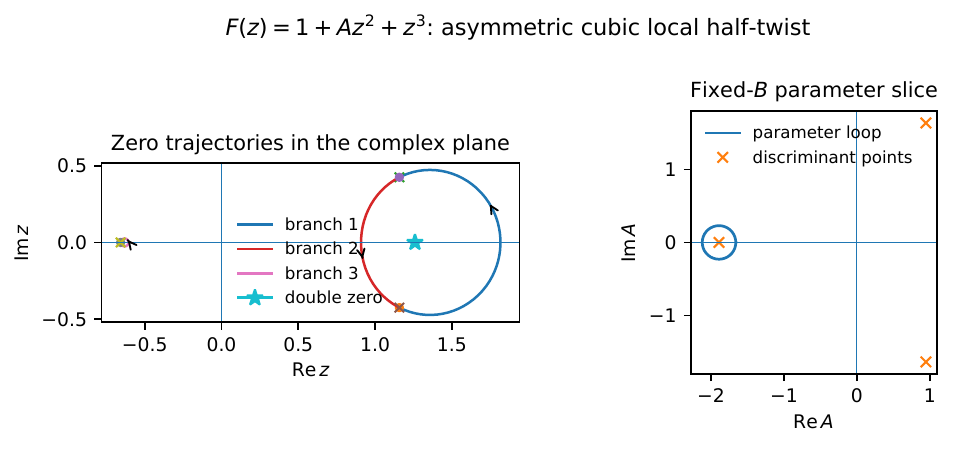}
  \caption{Asymmetric cubic half-twist for the trinomial slice $F_A(z)=1+Az^2+z^3$. The loop in the $A$-plane surrounds one point of the discriminant $A^3=-27/4$. Along this loop two stellar zeros perform a local exchange while the third remains separated. The example shows that sparse finite-Fock slices realize local braid generators without relying on the high symmetry of the cubic binomial family.}
  \label{fig:asymmetric-cubic-discriminant-loop}
\end{figure}

These examples illustrate the two complementary roles of the families introduced above. The full regular stratum contains all braid types by \cref{prop:realization}. Sparse Fock slices, by contrast, select explicit discriminant complements and produce concrete local braid generators that can be followed directly at the level of stellar zeros. This is the mechanism by which algebraic root monodromy becomes a geometric invariant of finite-Fock non-Gaussian state families.

\section{Cubic finite-Fock laboratory}
\label{sec:cubic-laboratory}

The general realization result of \cref{prop:realization} is useful, but it is important to see that it already has a concrete meaning in the experimentally motivated cubic family. In degree three, the regular stellar stratum is the complement of the discriminant hypersurface \eqref{eq:cubic-discriminant}, and its fundamental group is $B_3$. The following construction gives explicit loops of cubic finite-Fock states realizing the standard local exchanges of two stellar zeros.

Fix three distinct points $a,b,c\in\C$. Let
\begin{equation}
  m_{ab}=\frac{a+b}{2},
  \qquad
  \rho_{ab}=\frac{a-b}{2}.
\end{equation}
For $t\in[0,1]$, set
\begin{equation}
  z_1(t)=m_{ab}+\rho_{ab}e^{i\pi t},
  \qquad
  z_2(t)=m_{ab}-\rho_{ab}e^{i\pi t},
  \qquad
  z_3(t)=c.
  \label{eq:cubic-half-twist-roots}
\end{equation}
The unordered configuration at $t=1$ is the same as at $t=0$, but the two points $a$ and $b$ have been exchanged by a half-turn around their midpoint. The corresponding monic cubic is
\begin{equation}
  P_t(z)=\prod_{j=1}^3 \bigl(z-z_j(t)\bigr)
  =z^3+\alpha(t)z^2+\beta(t)z+\gamma(t).
  \label{eq:cubic-half-twist-polynomial}
\end{equation}
The coefficients $(\alpha(t),\beta(t),\gamma(t))$ define a loop in the regular cubic coefficient space because the three roots remain distinct throughout the motion.

\begin{proposition}[Cubic realization of the elementary half-twist]
\label{prop:cubic-half-twist}
The loop of cubic finite-Fock states associated with \eqref{eq:cubic-half-twist-polynomial} realizes a local half-twist exchanging the two stellar zeros initially located at $a$ and $b$, while the third zero remains fixed. With a labeling and projection convention, this is a conjugate of one of the standard generators $\sigma_1$ or $\sigma_2$ of $B_3$.
\end{proposition}

\begin{proof}
The path \eqref{eq:cubic-half-twist-roots} is a path in $\Conf_3(\C)$ whose endpoint differs from its initial point by the transposition of the first two labeled points. After quotienting by $S_3$, it becomes a loop in $\Conf_3(\C)/S_3$. The path in coefficient space \eqref{eq:cubic-half-twist-polynomial} is precisely the image of this loop under the Viete map. Since the motion is the standard local exchange of two points in a small disk disjoint from the third point, its braid class is a half-twist. The identification with $\sigma_1$ or $\sigma_2$ depends only on how the three zeros are initially ordered in the chosen braid projection.
\end{proof}

\begin{corollary}[Generators in the cubic finite-Fock family]
\label{cor:cubic-generators}
The experimentally motivated cubic finite-Fock family contains loops realizing a generating set of $B_3$.
\end{corollary}

\begin{proof}
For instance, take a base configuration with three distinct real zeros $x_1<x_2<x_3$. A half-turn exchange of the pair $(x_1,x_2)$ inside a small disk disjoint from $x_3$ gives one elementary half-twist; a half-turn exchange of $(x_2,x_3)$ inside a small disk disjoint from $x_1$ gives the other. In the usual projection convention these are the standard generators $\sigma_1$ and $\sigma_2$, up to the orientation convention for positive crossings. Since $\sigma_1$ and $\sigma_2$ generate $B_3$, the cubic finite-Fock family contains loops realizing a generating set.
\end{proof}

The corresponding finite-Fock amplitudes are obtained from \eqref{eq:cubic-normalized-coefficients} by reversing the normalization:
\begin{equation}
  c_3 \neq 0,
  \qquad
  c_2=\frac{\alpha c_3}{\sqrt{3}},
  \qquad
  c_1=\frac{\beta c_3}{\sqrt{6}},
  \qquad
  c_0=\frac{\gamma c_3}{\sqrt{6}}.
  \label{eq:cubic-coefficients-back}
\end{equation}
Thus the loop of roots in \eqref{eq:cubic-half-twist-roots} is not merely a loop of abstract cubics: through \eqref{eq:cubic-coefficients-back}, it is a loop of projective states in $\Proj(\Hil_3)$.

Because the construction keeps the three roots distinct for every $t$, the coefficient loop remains in the complement of the cubic discriminant. In this sense the cubic finite-Fock family is the smallest experimentally motivated finite-dimensional setting in which the two elementary generators of a nonabelian braid group are both visible.

\section{Finite stellar divisors and Gaussian covariance}
\label{sec:gaussian-covariance}

The finite-Fock chart is the most elementary algebraic realization of the stellar construction, but it is only a slice of the finite stellar-rank hierarchy. Gaussian unitaries may transform a finite Fock superposition into a state whose Bargmann representative is no longer a polynomial. The invariant object is therefore not the polynomial chart itself, but the finite divisor of stellar zeros.

This section makes that statement precise. It is aligned with recent stellar-decomposition approaches, where finite Fock support and Gaussian transformations are separated in the Bargmann picture \cite{Motamedi2026StellarDecomposition}. We first specify the admissible Gaussian factors in the Bargmann-Fock representation, then show that the Gaussian-extended finite-rank stratum projects onto the same configuration space of stellar zeros. The additional Gaussian parameters form a contractible fiber, so the braid topology remains carried by the stellar divisor.

\subsection{Admissible Gaussian factors and finite stellar divisors}

In the standard Bargmann-Fock space, a pure single-mode Gaussian factor can be written, up to an overall nonzero scalar, as
\begin{equation}
  E_{\tau,\mu}(z)
  =
  \exp\left(\frac{\tau}{2}z^2+\mu z\right),
  \qquad
  |\tau|<1,
  \quad
  \mu\in\C.
  \label{eq:admissible-gaussian-factor}
\end{equation}
The condition $|\tau|<1$ is the Bargmann-Fock normalizability condition for the quadratic part. A constant term in the exponent is omitted because we work projectively: multiplying the Bargmann representative by a nonzero scalar does not change the physical ray nor the zero divisor.

\begin{lemma}[Bargmann-Fock admissibility of quadratic Gaussian factors]
\label{lem:bargmann-admissibility}
The Gaussian factor $E_{\tau,\mu}$ belongs to the Bargmann-Fock space if and only if $|\tau|<1$. Multiplication by a polynomial does not change this quadratic admissibility condition.
\end{lemma}

\begin{proof}
In the Bargmann-Fock norm, the relevant weight is
\begin{equation}
  |E_{\tau,\mu}(z)|^2 e^{-|z|^2}
  =
  \exp\left(\operatorname{Re}(\tau z^2)+2\operatorname{Re}(\mu z)-|z|^2\right).
\end{equation}
If $|\tau|<1$, the quadratic form $-|z|^2+\operatorname{Re}(\tau z^2)$ is negative definite, so the linear term and any polynomial prefactor are dominated by a strictly decaying Gaussian. Conversely, if $|\tau|>1$, the exponent grows quadratically along directions where $\operatorname{Re}(\tau z^2)=|\tau||z|^2$. If $|\tau|=1$, the quadratic decay is lost along the corresponding asymptotic directions, and the Bargmann-Fock integral diverges. Thus the admissible parameter domain is the open unit disk.
\end{proof}

\begin{remark}[The boundary $|\tau|=1$]
\label{rem:tau-boundary}
The boundary $|\tau|=1$ is therefore not part of the finite-norm Bargmann-Fock stratum. It may be viewed as an ideal limiting boundary, for instance in an infinite-squeezing limit, but it is not included in the state space on which the braid monodromy is defined. In particular, the Gaussian parameter space used below is the open disk $\mathbb D$, not its closure.
\end{remark}

\begin{definition}[Admissible finite stellar-divisor stratum]
\label{def:gaussian-stellar-stratum}
Let $\mathcal G_r$ denote the set of projective pure single-mode states whose Bargmann representative admits a factorization
\begin{equation}
  F(z)=E_{\tau,\mu}(z)P(z),
  \qquad
  |\tau|<1,
  \quad
  \mu\in\C,
  \quad
  \deg P=r,
  \label{eq:admissible-gaussian-polynomial-factorization}
\end{equation}
where $P$ is a polynomial. Its regular part $\mathcal G_r^{\mathrm{reg}}$ consists of those states for which $P$ has $r$ simple zeros. The stellar-divisor map is
\begin{equation}
  \Div_\star:
  \mathcal G_r^{\mathrm{reg}}
  \longrightarrow
  \Conf_r(\C)/S_r,
  \qquad
  [F]\longmapsto Z(P).
  \label{eq:stellar-divisor-map-extended}
\end{equation}
\end{definition}

The exponential factor $E_{\tau,\mu}$ has no zeros, so the stellar divisor of $F$ is exactly the divisor of $P$, counted with multiplicity. The finite-Fock chart considered in the previous sections is the special algebraic slice $\tau=\mu=0$ with $P$ monic.

\begin{definition}[Stellar divisor with multiplicities]
\label{def:stellar-divisor-multiplicity}
If
\begin{equation}
  P(z)=\prod_{j=1}^m (z-z_j)^{m_j},
  \qquad
  \sum_{j=1}^m m_j=r,
\end{equation}
then the finite stellar divisor is
\begin{equation}
  D_\star(F)=\sum_{j=1}^m m_j[z_j].
\end{equation}
The regular stratum is the open part where $m_j=1$ for all $j$, so that $D_\star(F)$ is an unordered configuration of $r$ distinct points. The discriminant is precisely the locus where this divisor ceases to be reduced.
\end{definition}

\begin{lemma}[Uniqueness of the monic Gaussian-polynomial form]
\label{lem:unique-gaussian-polynomial}
Every projective class in $\mathcal G_r$ admits a representative of the form
\begin{equation}
  F(z)=E_{\tau,\mu}(z)P(z),
  \qquad
  |\tau|<1,
\end{equation}
with $P$ monic of degree $r$. In this normalization, the pair $(\tau,\mu)$ and the monic polynomial $P$ are uniquely determined by the projective class.
\end{lemma}

\begin{proof}
Starting from \eqref{eq:admissible-gaussian-polynomial-factorization}, divide by the leading coefficient of $P$ and absorb this nonzero scalar into the projective representative. This makes $P$ monic. Suppose now that
\begin{equation}
  E_{\tau,\mu}(z)P(z)
  =
  c\,E_{\tau',\mu'}(z)P'(z)
\end{equation}
for some $c\in\C^*$, with $P$ and $P'$ monic of degree $r$. Since the exponential factors are zero-free, $P$ and $P'$ have the same zero divisor. Being monic, they are equal. Hence
\begin{equation}
  \exp\left(\frac{\tau-\tau'}{2}z^2+(\mu-\mu')z\right)=c
\end{equation}
is constant. Therefore $\tau=\tau'$ and $\mu=\mu'$. The remaining constant is precisely the projective scalar.
\end{proof}

\subsection{Projection to the stellar configuration space}

Let
\begin{equation}
  \mathbb D=\bigl\{\tau\in\C: |\tau|<1\bigr\}.
\end{equation}
By \cref{lem:unique-gaussian-polynomial}, the regular Gaussian-extended finite-rank stratum can be written as a product of Gaussian parameters and a regular monic polynomial, or equivalently as Gaussian parameters plus a stellar configuration.

\begin{proposition}[Gaussian-extended stratum as a configuration fibration]
\label{prop:gaussian-fibration}
There is a natural identification
\begin{equation}
  \mathcal G_r^{\mathrm{reg}}
  \simeq
  \mathbb D\times\C\times\bigl(\Conf_r(\C)/S_r\bigr),
  \label{eq:gaussian-stratum-product}
\end{equation}
where $(\tau,\mu)$ are the Gaussian parameters and the last factor is the unordered stellar-zero configuration. Under this identification, the map $\Div_\star$ is the projection onto the configuration factor.
\end{proposition}

\begin{proof}
Given a class in $\mathcal G_r^{\mathrm{reg}}$, \cref{lem:unique-gaussian-polynomial} gives a unique representative $E_{\tau,\mu}P$ with $P$ monic. The roots of the monic polynomial define a point of $\Conf_r(\C)/S_r$. Conversely, given $(\tau,\mu)$ and an unordered configuration $\{z_1,\ldots,z_r\}$, define
\begin{equation}
  P_Z(z)=\prod_{j=1}^r (z-z_j).
\end{equation}
Then $E_{\tau,\mu}P_Z$ defines an element of $\mathcal G_r^{\mathrm{reg}}$. These two constructions are inverse to each other. The projection onto the configuration factor is exactly \eqref{eq:stellar-divisor-map-extended}.
\end{proof}

\begin{corollary}[Braid topology of the Gaussian-extended stratum]
\label{cor:gaussian-stratum-braid-group}
The Gaussian-extended regular finite-rank stratum has the same fundamental group as the finite-Fock regular stratum:
\begin{equation}
  \pi_1\bigl(\mathcal G_r^{\mathrm{reg}}\bigr)
  \simeq
  \pi_1\bigl(\Conf_r(\C)/S_r\bigr)
  \simeq
  B_r.
\end{equation}
Moreover, the braid monodromy of a loop in $\mathcal G_r^{\mathrm{reg}}$ is the braid monodromy of its projected stellar divisor.
\end{corollary}

\begin{proof}
The Gaussian parameter space $\mathbb D\times\C$ is contractible. Hence the product decomposition \eqref{eq:gaussian-stratum-product} has the same fundamental group as its configuration factor. The second statement follows because $\Div_\star$ is the projection onto that factor.
\end{proof}

\begin{corollary}[Trivial braid monodromy in Gaussian fibers]
\label{cor:gaussian-fiber-trivial}
Let $\Gamma:S^1\to \mathcal G_r^{\mathrm{reg}}$ be a loop whose stellar divisor is constant, equivalently a loop contained in a fiber of $\Div_\star$. Then $\Gamma$ has trivial braid monodromy. More generally, two loops in $\mathcal G_r^{\mathrm{reg}}$ whose projected divisor loops are homotopic in $\Conf_r(\C)/S_r$ define the same element of $B_r$.
\end{corollary}

\begin{proof}
Under the product decomposition \eqref{eq:gaussian-stratum-product}, $\Div_\star$ is the projection onto the configuration factor. A loop contained in a fiber has a constant projection, hence represents the identity element of $\pi_1(\Conf_r(\C)/S_r)\simeq B_r$. The second statement follows because braid monodromy is, by definition, the homotopy class of the projected divisor loop.
\end{proof}

\begin{remark}
This result is the precise sense in which the braid topology belongs to the stellar divisor rather than to the particular finite-Fock chart. The finite-Fock chart gives a canonical algebraic slice for computations; the Gaussian-extended stratum shows that adding admissible Gaussian factors does not add new fundamental-group topology. It only adds zero-free parameters in the fiber. In particular, a cyclic variation of the Gaussian parameters alone cannot create a nontrivial braid unless the stellar divisor itself moves nontrivially.
\end{remark}

\subsection{Gaussian covariance}

\begin{proposition}[Gaussian invariance of finite stellar rank]
\label{prop:gaussian-finite-rank}
Let $U_G$ be a single-mode Gaussian unitary. If a pure state has finite stellar rank $r$, then $U_G\ket{\psi}$ also has finite stellar rank $r$. Equivalently, if
\begin{equation}
  F_\psi(z)=E_{\tau,\mu}(z)P(z),
  \qquad \deg P=r,
\end{equation}
then the Bargmann representative of $U_G\ket{\psi}$ can be written as
\begin{equation}
  F_{U_G\psi}(z)=E_{\tau_G,\mu_G}(z)P_G(z),
  \qquad \deg P_G=r,
  \label{eq:gaussian-transformed-finite-rank}
\end{equation}
for some admissible Gaussian parameters $|\tau_G|<1$, $\mu_G\in\C$, and some polynomial $P_G$.
\end{proposition}

\begin{proof}
This is the single-mode Gaussian invariance of the stellar hierarchy established in the stellar representation of non-Gaussian states \cite{Chabaud2020}. In the Bargmann picture, finite stellar rank is precisely the existence of a representation by an admissible Gaussian factor times a polynomial of finite degree. Gaussian unitaries preserve that degree. The statement is also consistent with the holomorphic description of rank-preserving Gaussian evolutions, where Gaussian parameters evolve together with the polynomial zero data \cite{ChabaudMehraban2022}.
\end{proof}

\begin{corollary}[Gaussian families and braid monodromy]
\label{cor:gaussian-families-braid}
Let $\lambda\mapsto \ket{\psi_\lambda}$ be a continuous finite-rank family in $\mathcal G_r$, and let $\lambda\mapsto U_G(\lambda)$ be a continuous family of single-mode Gaussian unitaries. If the transformed stellar divisors
\begin{equation}
  \Div_\star\bigl(U_G(\lambda)\psi_\lambda\bigr)
\end{equation}
remain regular along a loop in parameter space, then the transformed family defines a braid monodromy in $B_r$.
\end{corollary}

\begin{proof}
By \cref{prop:gaussian-finite-rank}, the transformed family remains in finite stellar rank $r$. Regularity along the loop means that the $r$ zeros of the polynomial part $P_G$ are simple throughout the loop. Equivalently, the projected loop under $\Div_\star$ lies in $\Conf_r(\C)/S_r$. Its fundamental-group class is therefore an element of $B_r$.
\end{proof}

\begin{remark}[Regularity and discriminant crossings]
Gaussian covariance preserves the number of stellar zeros counted with multiplicity, but it does not by itself guarantee that a given path avoids the discriminant. A squeezing-driven path, for example, can move stellar zeros in a non-affine way and may pass through a point where two zeros collide. The braid monodromy is therefore attached to regular Gaussian-parametrized paths, i.e. to paths whose stellar divisors remain in $\Conf_r(\C)/S_r$.
\end{remark}

The previous statements are stronger than the affine covariance used in earlier versions of this draft: they treat the finite stellar divisor as the intrinsic object and not the finite-Fock polynomial chart. Nevertheless, two elementary Gaussian operations have a simple affine action on the zeros, which is useful for physical interpretation.

Let $h(z)=az+b$ be an affine biholomorphism of $\C$, with $a\neq0$, and let $G(z)$ be an entire function without zeros. Define
\begin{equation}
  (T_{G,h}F)(z)=G(z)F(h(z)).
  \label{eq:zero-free-affine-action}
\end{equation}

\begin{proposition}[Affine zero-free covariance]
\label{prop:affine-covariance}
Let $F=E_{\tau,\mu}P$ have a regular finite stellar divisor $Z(P)$. Then $T_{G,h}F$ has the same number of stellar zeros, all simple, and
\begin{equation}
  Z(T_{G,h}F)=h^{-1}\bigl(Z(P)\bigr).
  \label{eq:affine-zero-transport}
\end{equation}
Consequently, a regular family is sent to the braid monodromy of the affinely transported zero configuration.
\end{proposition}

\begin{proof}
Because $G$ and $E_{\tau,\mu}$ are zero-free, the zeros of $T_{G,h}F$ are exactly the solutions of $P(h(z))=0$. Since $h$ is a biholomorphism, this gives \eqref{eq:affine-zero-transport}. Simplicity is preserved because $h'(z)=a\neq0$. Applying this pointwise along a regular loop transports the corresponding path in configuration space by a homeomorphism.
\end{proof}

\begin{example}[Phase rotations and displacements]
In the Bargmann representation, the phase rotation $R(\theta)=e^{-i\theta a^\dagger a}$ acts as
\begin{equation}
  F(z)\longmapsto F(e^{-i\theta}z),
\end{equation}
so its zeros rotate by $z_j\mapsto e^{i\theta}z_j$. Similarly, the displacement operator $D(\alpha)$ acts by
\begin{equation}
  F(z)
  \longmapsto
  e^{-|\alpha|^2/2+\alpha z}F(z-\overline{\alpha}),
  \label{eq:bargmann-displacement}
\end{equation}
so its zeros translate by $z_j\mapsto z_j+\overline{\alpha}$. In both cases, the exponential prefactor has no zeros.
\end{example}

\begin{remark}[Squeezing is not merely affine transport]
The squeezing part of a general Gaussian unitary is better viewed as an evolution within the finite stellar-rank class $E_{\tau,\mu}P$ than as an affine transformation of the original zeros. In the holomorphic representation, Gaussian Hamiltonian evolution can produce nontrivial dynamics of the zeros together with an evolution of the Gaussian factor \cite{ChabaudMehraban2022}. This is precisely why the invariant is formulated at the level of stellar divisors: the finite-Fock chart gives an exact configuration-space model, while the Gaussian-extended finite-rank stratum shows that admissible Gaussian factors form a contractible zero-free fiber over the same configuration-space base.
\end{remark}

\section{Physical reconstruction and interpretation}
\label{sec:physical-reconstruction}

The stellar braid is not an expectation value of a single observable. It is a global invariant of a regular family of states, analogous in spirit to a geometric phase or a monodromy invariant. Nevertheless, in the pure finite-Fock setting considered here it is directly reconstructible from projective amplitudes, for instance after continuous-variable state tomography \cite{LvovskyRaymer2009}. This post-tomographic status is compatible with recent efforts to certify stellar-rank features and non-Gaussian witnesses experimentally \cite{Chabaud2021Certification,Provaznik2026Witnesses}. Experimental imperfections would turn braid extraction into a stability problem for inferred zeros; that statistical question is outside the scope of the present mathematical construction.

Suppose that a pure state is known, up to a global phase, in the truncated Fock basis
\begin{equation}
  \ket{\psi}=\sum_{n=0}^r c_n\ket{n},
  \qquad c_r\neq0.
\end{equation}
Then the projective amplitudes $[c_0:\cdots:c_r]$ determine the monic Bargmann polynomial
\begin{equation}
  P_\psi(z)
  =z^r+\frac{c_{r-1}\sqrt{r!}}{c_r\sqrt{(r-1)!}}z^{r-1}
  +\cdots+\frac{c_1\sqrt{r!}}{c_r}z
  +\frac{c_0\sqrt{r!}}{c_r}.
  \label{eq:reconstructed-monic-polynomial}
\end{equation}
Its zeros give the stellar configuration. Thus a loop of experimentally reconstructed amplitudes determines a loop of stellar configurations, provided the reconstructed path remains in the regular stratum.

\begin{proposition}[Post-tomographic reconstruction of stellar monodromy]
\label{prop:post-tomographic-reconstruction}
Let $t\mapsto [c_0(t):\cdots:c_r(t)]$ be a continuous loop of projective finite-Fock amplitudes with $c_r(t)\neq0$ and with simple stellar zeros for all $t$. Then the coefficients determine a well-defined braid class in $B_r$. This braid is invariant under continuous deformations of the reconstructed loop that do not cross the stellar discriminant.
\end{proposition}

\begin{proof}
The amplitudes determine the loop of monic polynomials \eqref{eq:reconstructed-monic-polynomial}, hence a loop in $\Sr$. Applying the stellar configuration map gives a loop in $\Conf_r(\C)/S_r$, whose fundamental-group class is an element of $B_r$. Homotopy invariance follows from the definition of the fundamental group, as long as the homotopy remains in the complement of the discriminant.
\end{proof}

This proposition clarifies the physical status of the invariant. A single state has a stellar configuration and a stellar rank. A loop or parameterized protocol of states has, in addition, a braid monodromy. The latter is stable under small perturbations of the protocol as long as no stellar-zero collision occurs. Therefore the braid detects how a preparation path winds around singular loci in the finite-Fock state space, rather than merely recording instantaneous properties of one state.

\begin{remark}[Complex state spaces and real preparation parameters]
\label{rem:real-parameter-families}
The coefficient spaces used above are complex because Bargmann representatives and Fock amplitudes are naturally complex. An experimental preparation, however, may depend on a real parameter manifold $M_{\mathbb R}$. Nothing in the construction requires the parameter space itself to be complex: a smooth map
\begin{equation}
  \Phi:M_{\mathbb R}\longrightarrow \mathcal G_r^{\mathrm{reg}}
\end{equation}
induces braid monodromy by the same projection to the stellar configuration space. The only regularity requirement is that the image avoid the discriminant and the degree-changing boundary. Thus the braid invariant is attached to real preparation loops just as well as to holomorphic or complex-algebraic families.
\end{remark}

\section{Local labelings and the sheaf viewpoint}

The ordered configuration space is a covering of the unordered configuration space:
\begin{equation}
  \pi:\Conf_r(\C)\longrightarrow \Conf_r(\C)/S_r.
  \label{eq:configuration-covering}
\end{equation}
Pulling this covering back by a stellar family
\begin{equation}
  \Zmap\circ\Phi:
  \mathcal M\longrightarrow \Conf_r(\C)/S_r
\end{equation}
gives a covering
\begin{equation}
  (\Zmap\circ\Phi)^*\Conf_r(\C)
  \longrightarrow
  \mathcal M.
  \label{eq:pullback-covering}
\end{equation}

\begin{proposition}[Local labeling sheaf]
\label{prop:local-labeling-sheaf}
Define a presheaf $\mathcal L_\Phi$ on $\mathcal M$ by
\begin{equation}
  \mathcal L_\Phi(U)
  =
  \left\{
  \text{continuous sections of }
  (\Zmap\circ\Phi)^*\Conf_r(\C)\longrightarrow U
  \right\}.
  \label{eq:labeling-sheaf}
\end{equation}
Then $\mathcal L_\Phi$ is the sheaf of local labelings of the stellar zeros. Its covering monodromy is the permutation part of the stellar braid monodromy, i.e. the composition
\begin{equation}
  \pi_1(\mathcal M,m_0)
  \stackrel{\rho_\Phi}{\longrightarrow}
  B_r
  \longrightarrow
  S_r.
\end{equation}
Moreover, $\mathcal L_\Phi$ admits a global section on a connected parameter space if and only if this permutation monodromy is trivial.
\end{proposition}

\begin{proof}
A point in the fiber of the pullback covering over $\lambda\in\mathcal M$ is an ordering of the unordered stellar configuration $\Zmap(\Phi(\lambda))$. Hence a local section is exactly a continuous local choice of labels for the zeros. Since coverings are locally trivial, such sections exist locally, and the usual gluing axiom for sections of a covering makes $\mathcal L_\Phi$ a sheaf.

Transporting a local labeling around a loop may return with a permuted ordering. This is precisely the monodromy of the covering, and it is obtained from the full braid monodromy by applying the natural projection $B_r\to S_r$. A global section is a labeling that is preserved under continuation along every loop; this is equivalent to trivial covering monodromy. The full braid monodromy may still have a nontrivial pure-braid component even when the permutation monodromy is trivial.
\end{proof}
\begin{remark}
This is the only sheaf-theoretic structure used in the present article. The associated sheaf of local sections records the local choices of stellar-zero labels and the obstruction to choosing them globally. More elaborate site- or topos-theoretic formulations may be useful for multimode states, singular strata, or families over more general parameter spaces, but they are not required for the finite-rank braid construction.
\end{remark}

\section{Mathematical scope of the construction}
\label{sec:mathematical-scope}

The construction has four layers, and the distinction between them prevents overinterpreting the examples. First, the full regular degree-$r$ finite-Fock stratum is equivalent to the unordered configuration space of $r$ points in the complex plane; at this level the full braid group $B_r$ is intrinsic and every braid is realized. Second, the Gaussian-extended finite-rank stratum adds admissible zero-free Gaussian factors, but these form a contractible fiber over the same configuration-space base, so the braid topology is still carried by the stellar divisor. Third, experimentally motivated finite-Fock subfamilies, such as the cubic family, give finite-dimensional physical laboratories inside this structure. Fourth, sparse trinomial slices give lower-dimensional algebraic sections where discriminants and local generators can be computed explicitly.

This separation is important. The statement that every braid is realizable belongs to the full regular stratum, or equivalently to its Gaussian extension through the divisor projection, not to every physical slice. Conversely, the trinomial discriminants are not claimed to exhaust the topology of the full stratum; they are exactly solvable probes of it. The physical content of the construction lies in the fact that the moving points are stellar zeros of non-Gaussian bosonic states, not arbitrary roots detached from a quantum state. The braid invariant is therefore an invariant of a regular preparation path or parameterized family of states.

\section{Conclusion}
\label{sec:conclusion}

The stellar representation turns finite-rank non-Gaussian states into holomorphic objects with finite zero divisors. The central point of this article is that this divisor carries more information than its degree. The degree gives the stellar rank; the braid monodromy records how the divisor is transported when a preparation path or parameterized family moves inside the regular finite-rank stratum.

The clean mathematical statement is that the regular degree-$r$ finite-Fock stellar stratum is biholomorphic to the unordered configuration space $\Conf_r(\C)/S_r$. Its fundamental group is therefore the Artin braid group $B_r$. Thus the braid group is not imported as an external analogy: it is the natural fundamental group of the regular stellar state space. The ordered stellar cover recovers the pure-braid exact sequence, while the local study of the discriminant identifies simple stellar-zero collisions with elementary half-twists.

The construction was then connected to finite-Fock state engineering in a way that keeps the hypotheses visible. The cubic family supported on $\ket{0},\ldots,\ket{3}$ gives a minimal nontrivial laboratory with direct experimental motivation. Sparse trinomial slices provide exactly solvable sections in which the discriminant can be written explicitly and concrete local braids can be visualized. The zero-trajectory figures are meant to emphasize the physical meaning of the invariant: the moving points are stellar zeros of non-Gaussian states, not roots of an abstract polynomial detached from a quantum system.

Finally, the Gaussian-extended finite-rank stratum shows that the braid topology belongs to the finite stellar divisor rather than to the choice of a strict polynomial chart. Writing Bargmann representatives as $E_{\tau,\mu}P$, with an admissible zero-free Gaussian factor and a monic polynomial, gives a projection onto the same configuration-space base with contractible Gaussian fiber. In this form the construction is compatible with Gaussian covariance while keeping the genuinely non-Gaussian information in the zero divisor.

The resulting invariant is best viewed as a robust, post-tomographic invariant of regular state-preparation loops. It complements Wigner negativity, stellar rank, approximate stellar rank, symplectic rank, and other scalar diagnostics because it does not measure an individual state alone; it measures the topology of a family. It is therefore not a replacement for scalar resource measures, but a record of the global motion of stellar divisors. This makes the construction a natural starting point for later dynamical problems in which propagation, dissipation, or medium parameters generate paths of stellar configurations. Mixed states, channels, and multimode zero varieties require additional structures and remain beyond the scope of this first construction.

\printbibliography

\end{document}